\documentclass[conference]{IEEEtran}
\IEEEoverridecommandlockouts

\usepackage{svg}
\usepackage{soul}
\usepackage{helvet}
\usepackage{cite}
\usepackage{hyperref}
\usepackage{amsmath,amssymb,amsfonts}
\usepackage{amsthm}
\usepackage{graphicx}
\usepackage{textcomp}
\usepackage{xcolor}
\usepackage{thmtools}
\usepackage{cleveref}
\usepackage{caption}
\usepackage{adjustbox}
\usepackage{subcaption}
\usepackage{color}
\usepackage[left=0.630in, right=0.640in, bottom=1.1in, top=1in, heightrounded]{geometry}
\usepackage{url,amssymb,threeparttable,multirow,booktabs, tabularx}
\def\BibTeX{{\rm B\kern-.05em{\sc i\kern-.025em b}\kern -.08em
    T\kern-.1667em\lower.7ex\hbox{E}\kern-.125emX}}
\usepackage{algorithm}
\usepackage{algpseudocode}
\usepackage{pifont}

\theoremstyle{plain}
\newtheorem{proposition}{Proposition}

\begin{document}
\title{Transformer-Based Sparse CSI Estimation for Non-Stationary Channels}
\author{
    \IEEEauthorblockN{
        Muhammad Ahmed Mohsin\IEEEauthorrefmark{1}, Muhammad Umer\IEEEauthorrefmark{1}, Ahsan Bilal\IEEEauthorrefmark{2}, Hassan Rizwan\IEEEauthorrefmark{3}, Sagnik Bhattacharya\IEEEauthorrefmark{1},\\Muhammad Ali Jamshed\IEEEauthorrefmark{4}, John M. Cioffi\IEEEauthorrefmark{1}
    }

    \IEEEauthorblockA{\IEEEauthorrefmark{1}Dept. of Electrical Engineering, Stanford University, Stanford, CA, USA}

    \IEEEauthorblockA{\IEEEauthorrefmark{2}Dept. of Computer Science, University of Oklahoma, OK, USA}

    \IEEEauthorblockA{\IEEEauthorrefmark{3}Dept. of Electrical \& Computer Engineering, University of California, Riverside, CA, USA}

    \IEEEauthorblockA{\IEEEauthorrefmark{4}Dept. of Electrical Engineering, University of Glasgow, Scotland, UK}

    \IEEEauthorblockA{Email: \{muahmed, mumer, sagnikb, cioffi\}@stanford.edu, ahsan.bilal-1@ou.edu, \\ muhammadali.jamshed@glasgow.ac.uk
    }
}

\maketitle

\begin{abstract}
Accurate and efficient estimation of Channel State Information (CSI) is critical for next-generation wireless systems operating under non-stationary conditions, where user mobility, Doppler spread, and multipath dynamics rapidly alter channel statistics. Conventional pilot-aided estimators incur substantial overhead, while deep learning approaches degrade under dynamic pilot patterns and time-varying fading. This paper presents a pilot-aided Flash-Attention Transformer framework that unifies model-driven pilot acquisition with data-driven CSI reconstruction through patch-wise self-attention and a physics-aware composite loss function enforcing phase alignment, correlation consistency, and time–frequency smoothness. Under a standardized 3GPP NR configuration, the proposed framework outperforms LMMSE and LSTM baselines by approximately 13\,dB in phase-invariant normalized mean-square error (NMSE) with markedly lower bit-error rate (BER), while reducing pilot overhead by $16\times$. These results demonstrate that attention-based architectures enable reliable CSI recovery and enhanced spectral efficiency without compromising link quality, addressing a fundamental bottleneck in adaptive, low-overhead channel estimation for non-stationary 5G and beyond-5G networks.
\end{abstract}


\section{Introduction}
Accurate acquisition of Channel State Information (CSI) is fundamental to realizing the full potential of Multiple-Input Multiple-Output Orthogonal Frequency Division Multiplexing (MIMO-OFDM) systems, enabling adaptive beamforming, link adaptation, and power allocation. In realistic wireless environments, however, channels exhibit non-stationary behavior due to user mobility, Doppler spread, and time-varying multipath propagation. These effects cause rapid variations in channel statistics, violating the quasi-stationary assumptions underlying conventional estimators. Classical pilot-aided techniques for CSI estimation struggle to maintain accuracy under such dynamics unless densely spaced pilot symbols are inserted, which results in substantial signaling overhead and degraded spectral efficiency. This trade-off between pilot density and estimation reliability is exacerbated in 5G and beyond-5G (B5G) systems operating over wide bandwidths and under high mobility. Consequently, developing channel estimators that are both robust and adaptive under non-stationary conditions remains a critical challenge.

Prior research has sought to reduce training overhead through pilot-aided estimation.~\cite{10107371} combined reliable data carriers with adaptive comb-type pilots to jointly optimize pilot spacing and reliability for improved mean-square error (MSE). However, pilot-aided schemes incur high overhead and adapt poorly when channel statistics drift. Recent works therefore employ deep learning (DL) for channel estimation, leveraging its ability to capture nonlinear propagation and generalize across diverse fading conditions. Lightweight fully connected networks have been used for real-time CSI recovery on edge devices~\cite{9430899}, while~\cite{9452036,10075639} unfolded iterative estimators into neural architectures combining data- and model-driven priors.~\cite{ngorima2025datapilotaidedtemporalconvolutional} proposed a temporal convolutional network (CNN) that jointly exploits pilot and remapped data subcarriers for vehicular fading, and other hybrid designs integrate pilot optimization through correlation-based embedding and neural pruning~\cite{sun2024pilotaidedjointtimesynchronization,9410430}. Beyond temporal adaptation, spatially non-stationary estimation for XL-MIMO and RIS employs polar-domain OMP~\cite{10631699}, HMM-aided visibility modeling~\cite{ceulemans2025nearfieldspatialnonstationarychannel}, and hierarchical sparsity learning~\cite{10286338,xu2025nearfieldpropagationspatialnonstationarity}. Probabilistic and generative frameworks such as Digital Twins~\cite{mohsin20256gtwinhybridgaussian} and diffusion-based estimators~\cite{mohsin2025conditionalpriorbasednonstationarychannel} further enable adaptive CSI reconstruction under non-stationary and multi-modal fading. Despite these advances, existing methods remain limited under dynamic channels: model-driven and unfolded networks degrade as statistics drift, CNNs cannot capture long-range time–frequency dependencies, and Transformer-based models are mostly confined to quasi-static settings without pilot adaptation. A unified, adaptive framework for non-stationary CSI recovery therefore remains largely unexplored.

This paper introduces a pilot-aided Flash-Attention Transformer framework for CSI reconstruction. The framework (i) unifies model-driven pilot acquisition with data-driven recovery via patch-wise self-attention and a physics-aware composite loss (phase-invariant fidelity, correlation alignment, and time–frequency smoothness), and (ii) explicitly conditions on the pilot mask to adapt to non-stationary Doppler and multipath evolution while capturing long-range time–frequency dependencies beyond convolutional or unfolded designs. Simulation results under a standardized 3GPP NR MIMO-OFDM setup demonstrate that the proposed model consistently surpasses LMMSE, LSTM, and LDAMP~\cite{8353153} baselines for CSI estimation across varying SNR levels, approaching the oracle dense-pilot benchmark. The framework achieves $\sim$13 dB gain in phase-invariant normalized mean-squared error (NMSE), with lower bit-error rate (BER) and subcarrier-wise error. Combined with a 16$\times$ pilot reduction, the proposed algorithm yields reliable CSI recovery and higher spectral efficiency without compromising link quality or robustness under non-stationary channels.

\section{System Model} \label{sec:system_model}
% We consider a 3GPP NR OFDM carrier with $N_{\mathrm{RB}}=64$ resource blocks (RBs) and subcarrier spacing $\Delta f=15\,\mathrm{kHz}$ with Normal cyclic prefix. Each RB contains $12$ subcarriers; hence the number of active subcarriers per OFDM symbol is $K = 12\,N_{\mathrm{RB}} = 768$. Per slot there are $L=14$ OFDM symbols. The transceiver employs a $(n_{\mathrm{Tx}},n_{\mathrm{Rx}})=(2,4)$ MIMO array with up to $n_{\ell}=\min(n_{\mathrm{Tx}},n_{\mathrm{Rx}})=2$ spatial layers.

\subsection{Channel Model and Pilot Configurations}
This work considers a 3GPP NR OFDM carrier characterized by $N_{\mathrm{RB}}$ resource blocks and subcarrier spacing $\Delta f$. Each resource block contains $12$ subcarriers, yielding $K = 12\,N_{\mathrm{RB}}$ active subcarriers per OFDM symbol. Each slot comprises $L$ OFDM symbols. The transceiver employs an $(n_{\mathrm{Tx}}, n_{\mathrm{Rx}})$ MIMO array supporting up to $n_{\ell} = \min(n_{\mathrm{Tx}}, n_{\mathrm{Rx}})$ spatial layers. The small-scale channel follows a tapped-delay-line profile (TDL-C) with RMS delay spread $\tau_{\mathrm{rms}}$ and maximum Doppler frequency $f_{\mathrm{D}}^{\max}$, modeling moderate mobility in an urban microcell street-canyon scenario. For transmit antenna $t$ and receive antenna $r$, the channel impulse response is
\begin{align}
h_{r,t}(\tau) &= \sum_{p=1}^{P} \alpha_{p,r,t}(t)\,\delta\!\big(\tau-\tau_{p}\big), \\[4pt]
\mathbb{E}\!\left[\alpha_{p,r,t}(t)\,\alpha_{p,r,t}^{\!*}(t+\Delta t)\right] 
&= \sigma_{p,r,t}^{2}\,J_{0}\!\big(2\pi f_{\mathrm{D},p}\Delta t\big),
\end{align}
with fixed excess delays $\{\tau_p\}$ and Bessel-type Doppler correlation (Clarke/Jakes model) for $0\!\le\!f_{\mathrm{D},p}\!\le\!f_{\mathrm{D}}^{\max}$. Large-scale effects include distance-dependent path loss and log-normal shadowing; the effective SNR per realization is adjusted accordingly. The nominal SNR prior to large-scale adjustment is $15\,\mathrm{dB}$.

Two sounding reference signal (SRS) configurations are employed:
\begin{itemize}
\item \textbf{Dense SRS (``perfect CSI''):} $\texttt{NumSRSSymbols}=4$ contiguous symbols per slot starting at $\ell_0=10$, frequency comb factor $K_{\mathrm{TC}}=2$ (comb-2 across subcarriers), periodic every slot, no frequency hopping.
\item \textbf{Sparse SRS (``noisy CSI''):} $\texttt{NumSRSSymbols}=1$ at $\ell_0=10$, comb factor $K_{\mathrm{TC}}=4$ (comb-4), periodic every $2$ slots.
\end{itemize}

Let $\Omega_{\mathrm{dense}}$ and $\Omega_{\mathrm{sparse}}$ denote the dense and sparse resource element (RE) index sets in $\{0,\dots,K-1\}\times\{0,\dots,L-1\}$, respectively, and let $\mathbf{M}\in\{0,1\}^{K\times L}$ be the binary pilot mask for the sparse case. Define $\mathbf{S}_{\mathrm{d}}$ as the unit-power pilot symbols on $\Omega_{\mathrm{dense}}$. The received pilots are
\begin{equation}
\mathbf{Y}_{\mathrm{d}} = \mathbf{P}_{\Omega_{\mathrm{dense}}}\big(\mathbf{H}\odot\mathbf{S}_{\mathrm{d}}\big) + \mathbf{W}_{\mathrm{d}},
\end{equation}
where $\mathbf{H}\!\in\!\mathbb{C}^{K\times L\times n_{\mathrm{Rx}}\times n_{\mathrm{Tx}}}$ represents the true channel tensor, $\odot$ denotes the Hadamard product, and $\mathbf{P}_{\Omega_{\mathrm{dense}}}$ is the projection operator onto $\Omega_{\mathrm{dense}}$. With oracle timing (perfect path-filter alignment), the simulator returns a dense, effectively noise-free estimate over the entire grid; this tensor is treated as the \emph{target} (``perfect CSI'').

For $\Omega_{\mathrm{sparse}}$ with pilots $\mathbf{S}_{\mathrm{s}}$, the channel estimate at pilot locations is
\begin{equation*}
\hat{\mathbf{H}}_{\Omega_{\mathrm{sparse}}} = \left.\mathbf{Y}_{\mathrm{s}}\oslash \mathbf{S}_{\mathrm{s}}\right|_{\Omega_{\mathrm{sparse}}},\;
\mathbf{Y}_{\mathrm{s}} = \mathbf{P}_{\Omega_{\mathrm{sparse}}}\big(\mathbf{H}\odot\mathbf{S}_{\mathrm{s}}\big) + \mathbf{W}_{\mathrm{s}},
\end{equation*}
followed by (i) symbol-wise temporal filling (pre-/post-hold around the single SRS symbol) and (ii) per-symbol frequency interpolation across subcarrier index $k$ using a linear operator $\mathcal{I}_\ell$ acting on the pilot subcarrier set $\{k_i\}$:
\begin{equation*}
\hat{H}[k,\ell] =
\begin{cases}
\mathcal{I}_\ell\!\left(\{(k_i, \hat{H}[k_i,\ell])\}\right), \quad \text{if } \exists\, (k_i,\ell)\in\Omega_{\mathrm{sparse}},\\[1mm]
\hat{H}[k,\ell^{-}], \quad \text{otherwise (temporal hold),}
\end{cases}
\end{equation*}
where $\ell^{-}$ denotes the most recent symbol containing pilot measurements. This yields $\hat{\mathbf{H}}\in\mathbb{C}^{K\times L\times n_{\mathrm{Rx}}\times n_{\mathrm{Tx}}}$ used as the \emph{input} (``noisy and interpolated'').

Stacking the RX and TX dimensions and defining the masking operator $\mathcal{M}(\mathbf{H})\triangleq \mathbf{M}\odot\mathbf{H}$, the sparse observation per slot satisfies
\begin{equation}
\mathbf{Y} = \mathcal{M}(\mathbf{H}) + \mathbf{W},\quad \mathbf{W}\sim\mathcal{CN}(\mathbf{0},\sigma^2\mathbf{I}),
\end{equation}
with a deterministic interpolator $\mathcal{G}$ such that $\hat{\mathbf{H}}=\mathcal{G}\big(\mathcal{M}(\mathbf{H})+\mathbf{W}\big)$. The dataset thus consists of paired tensors $(\hat{\mathbf{H}},\mathbf{H})$ over all slots and antenna pairs.

Let $|\Omega_{\mathrm{dense}}|$ and $|\Omega_{\mathrm{sparse}}|$ denote the number of pilot REs per slot. With comb-2 and $4$ pilot symbols,
\begin{equation}
\eta_{\mathrm{dense}} = \frac{|\Omega_{\mathrm{dense}}|}{K\,L}
\approx \frac{\tfrac{K}{2}\cdot 4}{K\,L}
= \frac{2}{L}
= \frac{2}{14},
\end{equation}
while for comb-4 and $1$ pilot symbol every $2$ slots, the average per-slot overhead is
\begin{equation}
\eta_{\mathrm{sparse}} = \frac{\tfrac{K}{4}\cdot 1}{2\,K\,L} = \frac{1}{8L}
= \frac{1}{112},
\end{equation}
illustrating a substantial reduction in pilot overhead at the cost of interpolation error that the Transformer model seeks to recover.

\subsection{Pilot Overhead and Achievable Rate} \label{sec:pilot_overhead_rate}
To quantify the spectral efficiency benefit of pilot reduction under reliable CSI reconstruction, the following proposition establishes the achievable rate gain.

\begin{proposition}[Rate Gain from Pilot Reduction under Reliable Reconstruction]
Consider a block-fading channel with coherence block length $T_c$ REs, pilot fraction $\alpha=\tau_p/T_c\in(0,1)$, per-RE SNR $\rho$, Rayleigh fading, and LMMSE channel estimation from $\tau_p$ orthogonal pilots. The standard training-based achievable rate per RE is
\begin{align}
R(\alpha)
&= (1-\alpha)\; \mathbb{E}\!\left[\log_2\!\big(1+\rho_{\mathrm{eff}}(\alpha)\,|h|^2\big)\right], \label{eq:Ralpha}\\
\rho_{\mathrm{eff}}(\alpha)
&\triangleq \frac{\rho\big(1-\sigma_e^2(\alpha)\big)}{1+\rho\,\sigma_e^2(\alpha)},\qquad
\sigma_e^2(\alpha)=\frac{1}{1+\rho\,\alpha T_c}.
\label{eq:rho_eff}
\end{align}
If a reconstruction mechanism enables a reduced pilot fraction $\alpha_1<\alpha_0$ while maintaining or improving the effective SNR, i.e., $\rho_{\mathrm{eff}}(\alpha_1)\ge \rho_{\mathrm{eff}}(\alpha_0)$, then
\begin{align}
R(\alpha_1)-R(\alpha_0)
&\ge (\alpha_0-\alpha_1)\,\mathbb{E}\!\left[\log_2\!\big(1+\rho_{\mathrm{eff}}(\alpha_1)\,|h|^2\big)\right] \nonumber\\
&\ge (\alpha_0-\alpha_1)\,\log_2\!\big(1+\rho_{\mathrm{eff}}(\alpha_1)\big),
\label{eq:gain_bound}
\end{align}
and the gain is strictly positive whenever $\alpha_1<\alpha_0$ and $\rho_{\mathrm{eff}}(\alpha_1)>0$.
\end{proposition}

\begin{proof}
Equation~\eqref{eq:Ralpha} represents the classical training-based lower bound with a pre-log factor $(1-\alpha)$ accounting for pilot overhead and an effective SNR $\rho_{\mathrm{eff}}$ that treats LMMSE estimation error as additional Gaussian noise. The error variance under $\tau_p$ pilots is $\sigma_e^2=(1+\rho\,\tau_p)^{-1}=(1+\rho\,\alpha T_c)^{-1}$, yielding~\eqref{eq:rho_eff}. Define $g(x)\!\triangleq\!\mathbb{E}\!\left[\log_2(1+x|h|^2)\right]$, which is nondecreasing and concave in $x$. Using~\eqref{eq:Ralpha},
\begin{multline*}
R(\alpha_1)-R(\alpha_0)=(1-\alpha_1)g\!\big(\rho_{\mathrm{eff}}(\alpha_1)\big)-(1-\alpha_0)g\!\big(\rho_{\mathrm{eff}}(\alpha_0)\big) \\
=(\alpha_0-\alpha_1)\,g\!\big(\rho_{\mathrm{eff}}(\alpha_1)\big)
-(1-\alpha_0)\Big(g\!\big(\rho_{\mathrm{eff}}(\alpha_0)\big)-g\!\big(\rho_{\mathrm{eff}}(\alpha_1)\big)\Big).
\end{multline*}
Under $\rho_{\mathrm{eff}}(\alpha_1)\!\ge\!\rho_{\mathrm{eff}}(\alpha_0)$, the bracketed term is nonpositive by monotonicity of $g$, which establishes the first inequality in~\eqref{eq:gain_bound}. For the second inequality, Jensen's inequality applied to the concave function $g$ with $\mathbb{E}[|h|^2]=1$ yields $g(x)\ge \log_2(1+x)$, thereby establishing the deterministic bound in~\eqref{eq:gain_bound}.
\end{proof}

\begin{figure*}[t]
 \centering
 \includegraphics[width=0.95\textwidth]{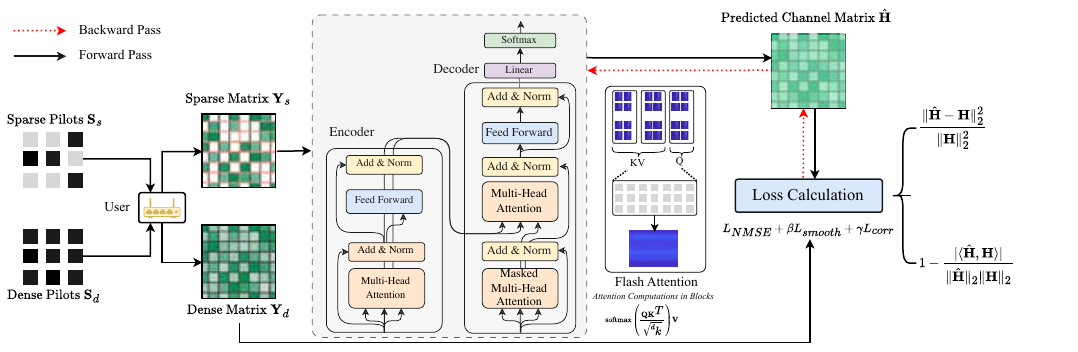}
\caption{Proposed pilot-aided Transformer framework for CSI reconstruction. Sparse pilots are first interpolated to form an initial estimate, which is then refined by the Transformer to recover high-fidelity CSI, followed by performance evaluation.}
 \label{fig:reward}
\end{figure*}

\section{Methodology}
The proposed framework integrates model-driven pilot acquisition with data-driven channel reconstruction to enable accurate CSI recovery under non-stationary propagation conditions. The framework follows a three-stage process comprising data preparation, model training, and performance evaluation. In the first stage, channel data are generated and preprocessed to provide training pairs of estimated and true CSI, denoted $\hat{\mathbf{H}}$ and $\mathbf{H}$, respectively. The second stage trains a Transformer model to map noisy pilot estimates to refined channel representations through end-to-end optimization. Finally, the trained model is evaluated using both reconstruction-based and link-level performance metrics to validate its effectiveness in improving channel estimation accuracy.

\subsection{Dataset Preparation}
The dataset used for training and evaluation consists of Monte Carlo realizations of MIMO-OFDM CSI generated under standardized 3GPP propagation conditions. Each sample contains a pair of tensors $(\hat{\mathbf{H}}, \mathbf{H})$, where $\mathbf{H} \in \mathbb{C}^{K \times L \times n_{\mathrm{Rx}} \times n_{\mathrm{Tx}}}$ represents the true channel tensor with $K$ subcarriers, $L$ OFDM symbols, $n_{\mathrm{Rx}}$ receive antennas, and $n_{\mathrm{Tx}}$ transmit antennas, while $\hat{\mathbf{H}}$ denotes the corresponding noisy and interpolated CSI obtained from sparse pilot-based estimation as described in Section~\ref{sec:system_model}.

\subsection{Model Architecture}
The proposed Flash-Attention Transformer is a compact yet expressive encoder-decoder network designed to model the joint time-frequency correlations of wireless channels. The input tensor $\hat{\mathbf{H}}$ is decomposed into its real and imaginary components to form a two-channel real-valued feature map $\mathbf{X}_0 \in \mathbb{R}^{K \times D_0}$, where $K$ denotes the number of subcarriers and $D_0$ represents the number of real-valued features per subcarrier after projection. A convolutional patch embedding layer maps each local region $\mathbf{x}_p$, representing a small time-frequency segment of the real and imaginary channel feature maps derived from $\hat{\mathbf{H}}$, into a latent vector $\mathbf{z}_p \in \mathbb{R}^{D_0}$ defined as $\mathbf{z}_p = \mathbf{W}_e \mathbf{x}_p + \mathbf{b}_e,$ where $\mathbf{W}_e$ and $\mathbf{b}_e$ are learnable projection parameters and $p$ denotes the patch index. Positional encodings $\mathbf{E}_{\mathrm{pos}}$ are then added to retain the frequency-time ordering of tokens, yielding the embedded sequence $\mathbf{Z}_0 = \{\mathbf{z}_p + \mathbf{E}_{\mathrm{pos},p}\}_{p=1}^{P}$, where $P$ denotes the total number of patches. This sequence is then processed by $N$ stacked Transformer encoder layers, each comprising a multi-head self-attention module and a feedforward subnetwork with residual connections and layer normalization. Within each layer, the attention mechanism adaptively models dependencies among channel tokens, generating refined representations $\mathbf{Z}_\ell$ that capture both short- and long-range variations across subcarriers and OFDM symbols, allowing the network to learn dynamic channel behavior and cross-subcarrier coupling without relying on fixed interpolation rules. Finally, a transposed convolutional decoder $\mathcal{F}_{\mathrm{dec}}(\cdot)$ upsamples $\mathbf{Z}_N$ to yield the reconstructed CSI estimate $\hat{\mathbf{H}} = \mathcal{F}_{\mathrm{dec}}(\mathbf{Z}_N)$. The model realizes a nonlinear mapping $\hat{\mathbf{H}} = f_{\Theta}(\hat{\mathbf{H}})$ trained end-to-end to minimize the composite loss in Section~\ref{sec:loss_formulation}.

\subsection{Channel Estimation via Spatio-Temporal Interpolation}
In contrast to traditional pilot-aided estimation applying explicit interpolation or filtering rules, the proposed approach enables the Transformer to \textit{implicitly} learn spatiotemporal relationships from data. Instead of manually defining low-pass, Wiener, or Kronecker filters, the model learns a nonlinear mapping $\hat{\mathbf{H}} = f_{\Theta}(\hat{\mathbf{H}})$ jointly capturing correlations across subcarriers, OFDM symbols, and antenna elements. The attention mechanism adaptively weighs dependencies to infer missing pilot information while accounting for non-stationary Doppler and angular variations. This unified formulation performs denoising and reconstruction in a single end-to-end step, recovering physically consistent CSI under dynamic, frequency-selective, and spatially non-stationary propagation conditions.

\subsection{Loss Formulation}
\label{sec:loss_formulation}
To ensure physically consistent and scale-robust channel reconstruction, the Transformer model is trained using a composite loss function that integrates normalized reconstruction fidelity, complex-phase alignment, and time–frequency smoothness regularization. The overall objective function is expressed as
\begin{equation}
    \mathcal{L}_{\mathrm{total}} 
    = \mathcal{L}_{\mathrm{pri}} 
    + \beta \mathcal{L}_{\mathrm{smooth}}
    + \gamma \mathcal{L}_{\mathrm{corr}},
    \label{eq:total_loss}
\end{equation}
where $\mathcal{L}_{\mathrm{pri}}$ denotes the primary reconstruction loss, $\mathcal{L}_{\mathrm{smooth}}$ enforces physical continuity across time and frequency, and $\mathcal{L}_{\mathrm{corr}}$ promotes phase-aligned correlation between the estimated and true channel responses. The weights $\beta$ and $\gamma$ are empirically selected to balance numerical accuracy and physical regularity.

\paragraph*{1) Primary Reconstruction Loss}
The core metric guiding training is the normalized mean square error (NMSE), defined as
\begin{equation}
    \mathcal{L}_{\mathrm{NMSE}} 
    = 
    \frac{\| \hat{\mathbf{H}} - \mathbf{H} \|_2^2}
         {\| \mathbf{H} \|_2^2},
    \label{eq:nmse_loss}
\end{equation}
where $\hat{\mathbf{H}}$ and $\mathbf{H}$ represent the predicted and true channel tensors, respectively. NMSE normalizes the squared reconstruction error by the instantaneous channel power, enabling fair optimization across varying SNR regimes. For channels with residual phase or scaling offsets, a phase-invariant variant is employed:
\begin{equation}
    \mathcal{L}_{\mathrm{SP\text{-}NMSE}}
    =
    \frac{\| \hat{\mathbf{H}} - \alpha^{*}\mathbf{H} \|_2^2}
         {\| \mathbf{H} \|_2^2}, \qquad
    \alpha^{*} = \frac{\langle \hat{\mathbf{H}}, \mathbf{H} \rangle}
                       {\| \mathbf{H} \|_2^2},
    \label{eq:spnmse_loss}
\end{equation}
where $\alpha^{*}$ aligns the predicted and reference channels through a complex scaling factor, thereby compensating for global phase and amplitude ambiguities commonly introduced by pilot-based estimation.

\paragraph*{2) Correlation Alignment Loss}
To further encourage structural coherence between the predicted and ground-truth CSI tensors, a complex correlation loss is employed:
\begin{equation}
    \mathcal{L}_{\mathrm{corr}} = 
    1 - 
    \frac{|\langle \hat{\mathbf{H}}, \mathbf{H} \rangle|}
         {\| \hat{\mathbf{H}} \|_2 \, \| \mathbf{H} \|_2}.
    \label{eq:corr_loss}
\end{equation}
This term penalizes decorrelation in both magnitude and phase, driving the model toward reconstructions that preserve relative antenna and subcarrier dependencies.

\paragraph*{3) Time–Frequency Smoothness Regularization}
Given that realistic wireless channels exhibit continuous variation across adjacent OFDM symbols and subcarriers, a smoothness regularization term is imposed:
\begin{equation}
    \mathcal{L}_{\mathrm{smooth}} 
    = 
    \lambda_t \| \nabla_t \hat{\mathbf{H}} \|_2^2 
    + 
    \lambda_f \| \nabla_f \hat{\mathbf{H}} \|_2^2,
    \label{eq:smooth_loss}
\end{equation}
where $\nabla_t$ and $\nabla_f$ denote discrete gradient operators along time and frequency dimensions, respectively. This term stabilizes learning and encourages physically plausible CSI surfaces by penalizing abrupt fluctuations caused by noisy pilot observations.

The composite formulation in~\eqref{eq:total_loss} jointly minimizes reconstruction distortion, misalignment, and local irregularity. Empirically, the phase-invariant NMSE dominates optimization, while the correlation and smoothness terms act as regularization priors to guide the Transformer's attention dynamics. This combination yields channel estimates that are both numerically accurate and physically consistent.

\begin{figure}[t]
  \centering
  \begin{subfigure}[t]{0.37\textwidth}
    \centering
    \includegraphics[width=\linewidth]{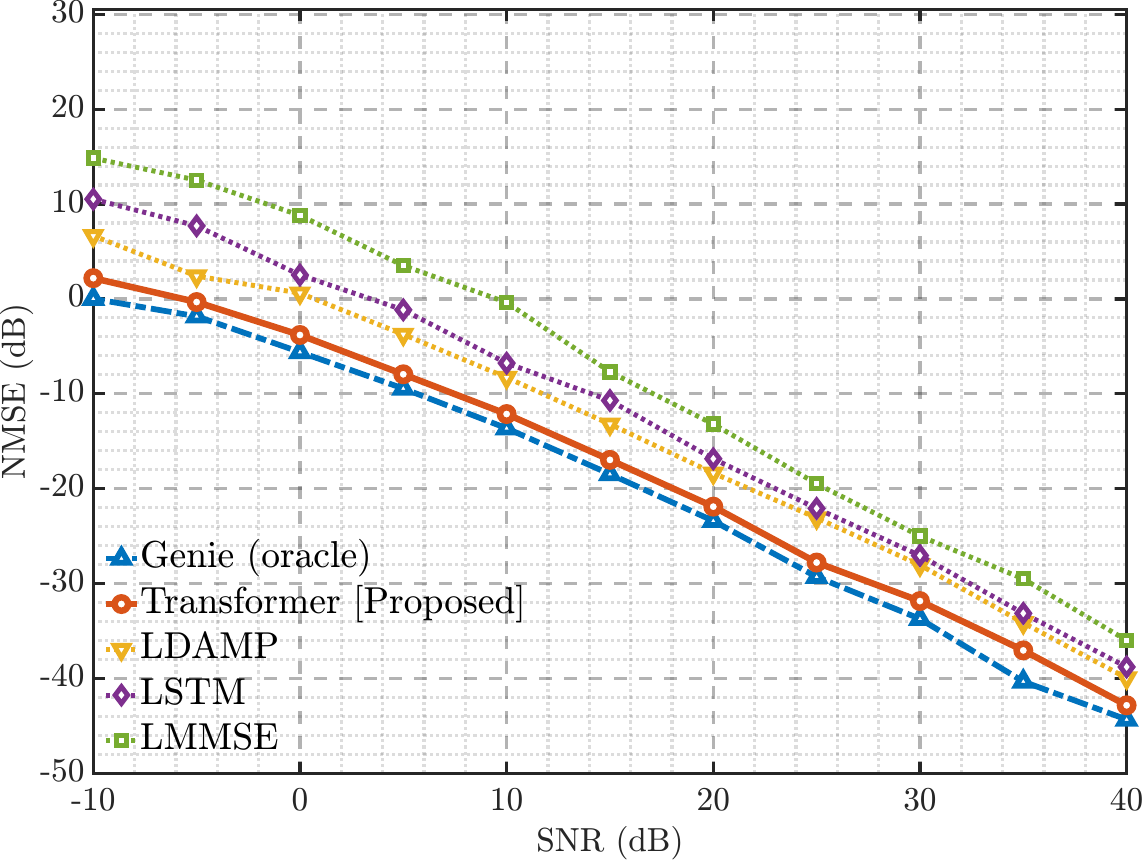}
    \caption{NMSE of the proposed model vs.\ LMMSE, LDAMP, LSTM, and oracle across SNR.}
    \label{fig:nmse_results}
  \end{subfigure}

  \vspace{1em}
  \begin{subfigure}[t]{0.37\textwidth}
    \centering
    \includegraphics[width=\linewidth]{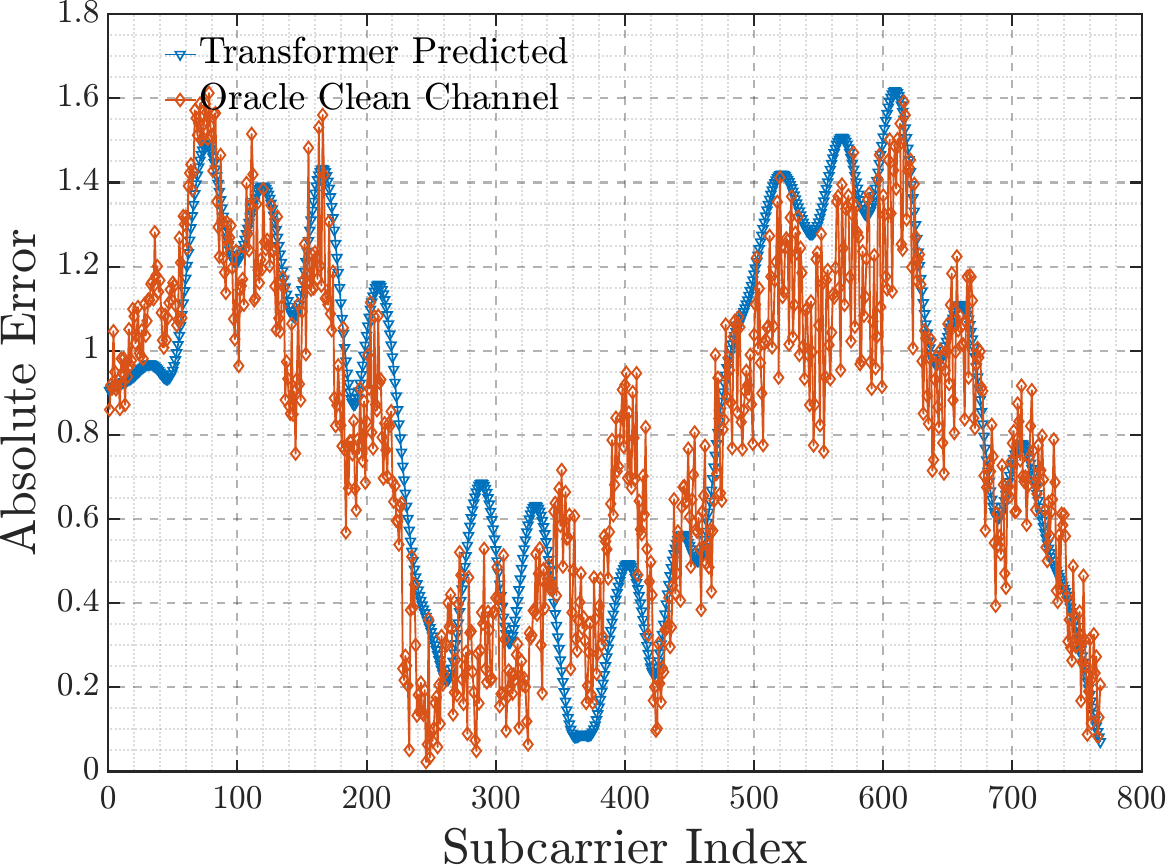}
    \caption{Absolute reconstruction error across subcarriers compared to the clean channel.}
    \label{fig:subcarrier_error}
  \end{subfigure}
  \caption{NMSE performance and subcarrier-wise accuracy of the Transformer estimator versus baselines across SNRs.}
  \label{fig:perf_two_panel}
\end{figure}

\begin{figure*}[t!]
     \centering
    \begin{subfigure}[t]{0.31\textwidth}
         \centering
         \includegraphics[width=\textwidth]{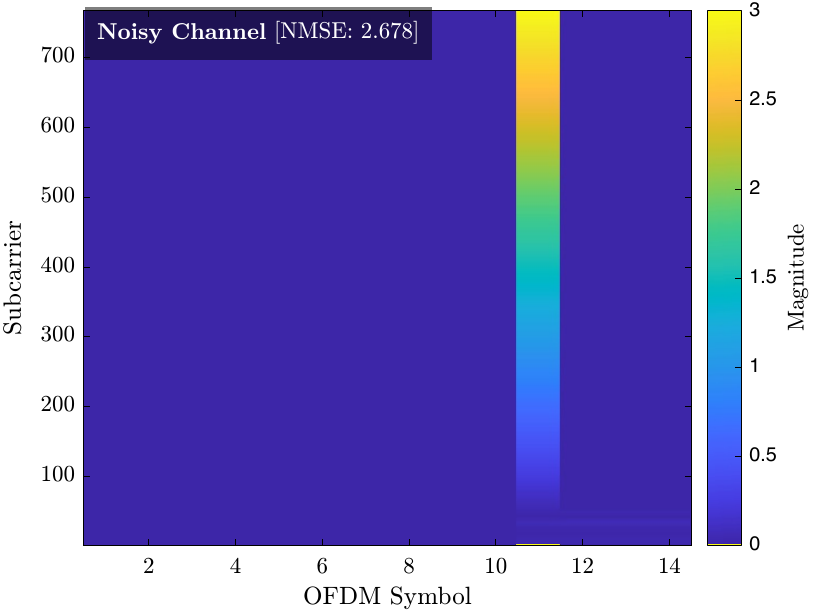}
         \caption{Noisy CSI obtained from sparse pilots after temporal hold and linear interpolation.}
         \label{fig:heatmap_noisy}
     \end{subfigure}
     \hspace{1em}
     \begin{subfigure}[t]{0.31\textwidth}
         \centering
         \includegraphics[width=\textwidth]{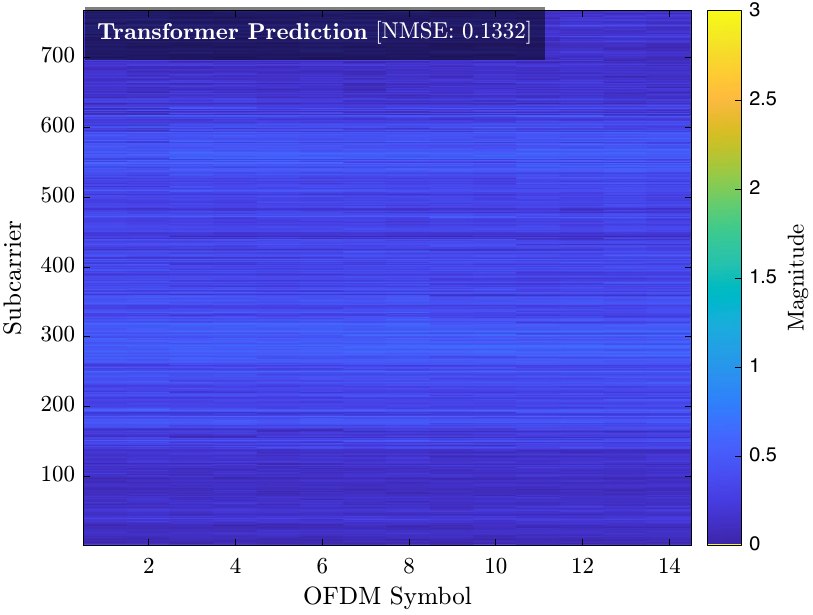}
         \caption{Reconstructed CSI from the proposed Transformer model.}
         \label{fig:heatmap_predicted}
     \end{subfigure}
     \hspace{1em}
     \begin{subfigure}[t]{0.31\textwidth}
         \centering
         \includegraphics[width=\textwidth]{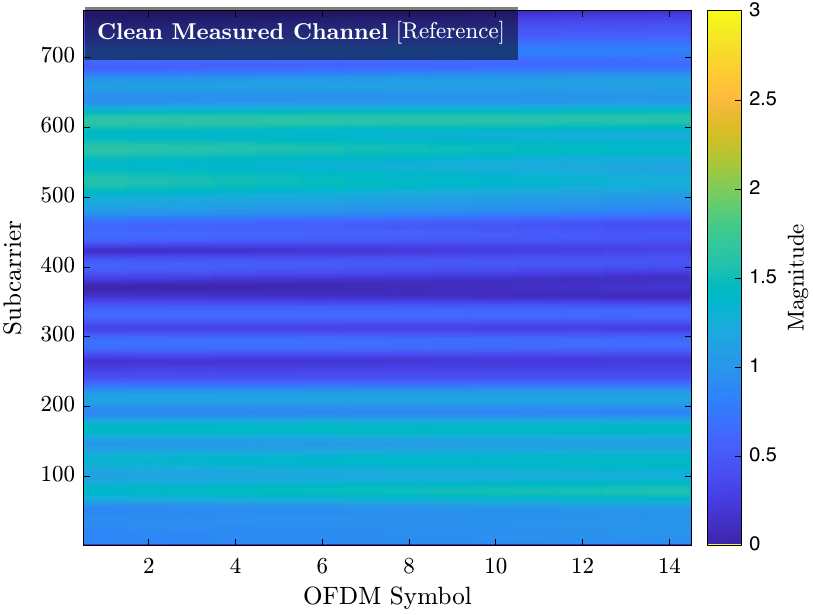}
         \caption{Ground-truth CSI obtained from dense sounding (oracle reference).}
         \label{fig:heatmap_truth}
     \end{subfigure}
      \caption{Time–frequency magnitude heatmaps of CSI: (a) noisy and interpolated sparse-pilot estimate, (b) Transformer-reconstructed CSI, and (c) oracle dense-pilot reference.}
     \label{fig:heatmaps}
\end{figure*}

\subsection{Evaluation Flow}
Model evaluation is performed in two complementary stages to quantify both reconstruction fidelity and communication performance. In the first stage, the trained Transformer is applied to unseen test data to reconstruct the refined estimate from the noisy pilot-aided inputs $\hat{\mathbf{H}}$. The NMSE in decibels is computed across all test realizations as
\begin{equation}
    \mathrm{NMSE}_{\mathrm{dB}} = 10 \log_{10} 
    \left( 
        \frac{\mathbb{E}[\|\hat{\mathbf{H}} - \mathbf{H}\|_2^2]}
             {\mathbb{E}[\|\mathbf{H}\|_2^2]}
    \right),
    \label{eq:nmse_eval}
\end{equation}
to measure estimation accuracy relative to the true channel. This metric reflects the model's ability to generalize across unseen fading conditions and quantifies reconstruction quality independently of absolute channel gain.

\section{Simulation Results}
\label{sec:results}

\subsection{Baselines and Metrics}
The proposed pilot-aided framework is compared against four widely used channel estimation methods. The first baseline is LMMSE with linear interpolation, where the received pilot vector is modeled as $\mathbf{y} = \mathbf{S}\mathbf{h} + \mathbf{w}$ with $\mathbf{w} \sim \mathcal{CN}(\mathbf{0}, \sigma^2\mathbf{I})$. The LMMSE estimate is obtained as $\hat{\mathbf{h}}_{\Omega} = \mathbf{R}_{hy}\mathbf{R}_{yy}^{-1}\mathbf{y}$, where $\mathbf{R}_{yy} = \mathbf{S}\mathbf{R}_{hh}\mathbf{S}^{H} + \sigma^2\mathbf{I}$. This is followed by temporal hold (zero-order hold across non-piloted OFDM symbols) and per-symbol linear frequency interpolation over non-pilot resource elements~\cite{gong2024channel}. The second baseline, LDAMP~\cite{8353153}, refines the time–frequency domain CSI iteratively by treating interpolation as a linear sensing model $\mathbf{A}$. The iterative updates are given by
\begin{align}
\mathbf{z}^{t} &= \mathbf{y} - \mathbf{A}\mathbf{x}^{t} + \frac{1}{\delta}\mathbf{z}^{t-1}\langle \mathcal{D}'_{\theta_t}(\mathbf{x}^{t} + \mathbf{A}^{H}\mathbf{z}^{t}) \rangle, \\
\mathbf{x}^{t+1} &= \mathcal{D}_{\theta_t}(\mathbf{x}^{t} + \mathbf{A}^{H}\mathbf{z}^{t}),
\end{align}
where $\delta$ denotes the state-evolution factor and $\mathcal{D}_{\theta_t}$ represents a learned denoising operator. The third baseline employs a complex-valued two-layer LSTM that processes the frequency-domain sequence of each OFDM symbol using
\begin{align}
\mathbf{h}_t, \mathbf{c}_t &= \mathrm{LSTM}(\mathbf{x}_t, \mathbf{h}_{t-1}, \mathbf{c}_{t-1}), \\
\hat{\mathbf{y}}_t &= \mathbf{W}_o \mathbf{h}_t + \mathbf{b}_o,
\end{align}
effectively capturing short-term temporal correlations~\cite{8353153}. Finally, a genie-aided oracle with comb-2 dense sounding and perfect timing provides an effectively noise-free channel $\mathbf{H}$, serving as the upper bound for comparison, i.e., $\hat{\mathbf{H}}_{\mathrm{genie}} = \mathbf{H}$. All baselines are evaluated under the 3GPP NR configuration with simulation parameters summarized in Table~\ref{tab:params}.

\begin{table}[t!]
\caption{Simulation parameters for 3GPP NR MIMO-OFDM.}
\label{tab:params}
\centering
\small
\begin{tabular}{l l}
\toprule
\textbf{Parameter} & \textbf{Value} \\
\midrule
Carrier grid & $N_{\mathrm{RB}}{=}64$ ($K{=}768$ subcarriers), $\Delta f{=}15\,\mathrm{kHz}$ \\
MIMO & $n_{\mathrm{Tx}}{=}2$, $n_{\mathrm{Rx}}{=}4$, $n_\ell{=}2$ \\
Channel & TDL-C, $\tau_{\mathrm{rms}}{=}251\,\mathrm{ns}$, $f_{\mathrm{D}}^{\max}{=}50\,\mathrm{Hz}$ \\
SNR (nominal) & $15\,\mathrm{dB}$ prior to large-scale adjustment \\
Dense SRS & $4$ symbols at $\ell_0{=}10$, comb-2, every slot \\
Sparse SRS & $1$ symbol at $\ell_0{=}10$, comb-4, every $2$ slots \\
Noise & AWGN per RE, $\mathcal{CN}(0,\sigma^2)$ \\
\bottomrule
\end{tabular}
\end{table}

\subsection{Results}
The proposed framework demonstrates consistently superior reconstruction accuracy across all SNR regimes, closely approaching the oracle upper bound. The NMSE curves in Fig.~\ref{fig:nmse_results} reveal the most pronounced gains at moderate SNR levels, where sparse pilots and interpolation artifacts most severely degrade conventional estimators. This demonstrates the ability to capture long-range time–frequency dependencies that LMMSE oversmoothes and that LSTM or LDAMP fail to model effectively. Fig.~\ref{fig:subcarrier_error} confirms that the proposed framework preserves fine spectral details and minimizes distortion near deep fades.

Time–frequency heatmaps in Fig.~\ref{fig:heatmaps} provide a qualitative comparison: the sparse-pilot input exhibits discontinuities and block-like artifacts, whereas the Transformer-based reconstruction restores structural smoothness, faithfully recovering the underlying CSI surface. Quantitatively, the phase-invariant NMSE decreases from $2.67$ to $0.13$ (from $4.26$\,dB to $-8.86$\,dB), representing approximately $13$\,dB improvement. Using this error level within the effective SNR formulation of Section~\ref{sec:pilot_overhead_rate} indicates that the reconstructed CSI yields link performance nearly identical to the dense-pilot case, confirming high CSI fidelity with drastically reduced pilot overhead.

With dense and sparse pilot fractions of $\alpha_0 = \tfrac{2}{14}$ and $\alpha_1 = \tfrac{1}{112}$, respectively, the reliable reconstruction satisfies the rate-gain condition of~\eqref{eq:gain_bound}, implying that pilot reduction directly translates into higher achievable throughput without sacrificing estimation quality. Among baselines, LMMSE deteriorates under sparse sampling and non-stationary fading, LDAMP fails to track rapidly varying time–frequency correlations, and LSTM lacks global temporal-spectral context. In contrast, The proposed architecture consistently outperforms all baselines by leveraging self-attention to capture long-range dependencies and employing physics-aware regularization to enforce temporal and spectral consistency, enabling accurate, stable, and low-overhead CSI recovery under highly non-stationary and pilot-sparse conditions.

\section{Conclusion}
% This paper introduces a pilot-aided transformer framework that fundamentally rethinks CSI recovery under sparse training, shifting from conventional local interpolation or model-unfolded estimation toward a global, context-aware reconstruction paradigm. By treating pilots as high-value structural anchors rather than dense sampling points, the proposed method leverages self-attention to infer long-range time–frequency dependencies that are inaccessible to classical estimators. Conceptually, this aligns with the theoretical insight established in Sec.~\ref{sec:pilot_overhead_rate}, if reliable reconstruction preserves effective SNR under reduced pilot fractions, spectral efficiency can improve without compromising estimation quality. The integration of phase-invariant error minimization, correlation alignment, and physical smoothness priors enables the model to remain consistent with channel physics while retaining the expressive power of attention-based learning. Overall, this approach demonstrates a promising direction for future wireless systems, where intelligence, rather than pilot density, becomes the key enabler of accurate and adaptive CSI acquisition in non-stationary environments.

This paper introduces a pilot-aided Transformer framework that redefines CSI recovery under sparse training, moving beyond conventional local interpolation or model-unfolded estimation toward a global, context-aware reconstruction paradigm. By treating pilots as high-value structural anchors rather than dense sampling points, the proposed approach leverages self-attention to infer long-range time–frequency dependencies that are inaccessible to classical estimators. Under a standardized 3GPP NR setup, the proposed model outperforms LMMSE, LSTM, and LDAMP baselines, achieving a $\sim$13\,dB gain in phase-invariant normalized mean-square error (NMSE) and significantly lower bit-error rate (BER) while using $16\times$ fewer pilots. The integration of phase-invariant error minimization, correlation alignment, and time–frequency smoothness regularization ensures physical consistency while maintaining the expressive capacity of attention-based learning. The proposed algorithm enables reliable CSI recovery and higher spectral efficiency without sacrificing link quality, addressing a key bottleneck in adaptive, low-overhead channel estimation for non-stationary 5G/B5G networks.

\footnotesize
\bibliographystyle{IEEEtran}
\bibliography{main}
\end{document}